\documentclass{article}

\usepackage{arxiv}

\usepackage[utf8]{inputenc} 
\usepackage[T1]{fontenc}    
\usepackage{hyperref}       
\usepackage{url}            
\usepackage{booktabs}       
\usepackage{amsfonts}       
\usepackage{nicefrac}       
\usepackage{microtype}      
\usepackage{lipsum}

\usepackage{amsmath}
\usepackage{multirow}

\newtheorem{theorem}{Theorem}
\newtheorem{proof}{Proof}

\title{New classes of tests for the Weibull distribution using Stein's method in the presence of random right censoring}

\author{
  E. ~Bothma\\
  Subject Group Statistics\\
  North-West University\\
  South Africa \\

   \And
  J.S. ~Allison\\
  Subject Group Statistics\\
  North-West University\\
  South Africa \\

     \And
  I.J.H. ~Visagie\\
  Subject Group Statistics\\
  North-West University\\
  South Africa \\
  \texttt{jaco.visagie@nwu.ac.za} \\
}

\begin{document}
\maketitle

\begin{abstract}
We develop two new classes of tests for the Weibull distribution based on Stein’s method. The proposed tests are applied in the full sample case as well as in the framework of random right censoring. We investigate the finite sample performance of the new tests using a comprehensive Monte Carlo study. In both the absence and presence of censoring, it is found that the newly proposed classes of tests outperform competing tests against the majority of the distributions considered. In the cases where censoring is present we consider various censoring distributions. Some remarks on the asymptotic properties of the proposed tests are included. The paper presents another result of independent interest; the test initially proposed in \cite{krit2014goodness} for use with full samples is amended to allow for testing for the Weibull distribution in the presence of censoring. The techniques developed in the paper are illustrated using two practical examples. In the first, we consider the survival times of patients with a certain type of leukemia. The second example is concerned with the initial remission times of leukemia patients, where the observed remission times are subject to random right censoring. We further include some concluding remarks along with avenues for future research.
\end{abstract}


\section{Introduction}
\label{Intro}
The Weibull distribution is often used in survival analysis as well as reliability theory, see e.g., \cite{kalbfleisch2011statistical}. This flexible distribution is a popular model which allows for constant, increasing and decreasing hazard rates. The Weibull distribution is also frequently applied in various engineering fields, including electrical and industrial engineering to represent, for example, manufacturing times, see \cite{jiang2011study}. As a result of its wide range of practical uses, a number of goodness-of-fit tests have been developed for the Weibull distribution; see, for example, \cite{mann1973men}, \cite{tiku1981testing}, \cite{liao1999new}, \cite{cabana2005using} as well as \cite{krit2014goodness}.

The papers listed above deal with testing for the Weibull distribution in the full sample case; i.e., where all lifetimes are observed. However, random right censoring often occurs in the fields mentioned above. For example, we may study the duration that antibodies remain detectable in a patient's blood after receiving a specific type of Covid-19 vaccine, i.e. the duration of the protection that the vaccine affords the recipient. When gathering the relevant data, we will likely not be able to measure this duration in all of the patients. For example, some may leave the study by emigrating to a different country while still having detectable antibodies. In this case, the exact time of interest is not observed. This situation is referred to as random right censoring, see e.g., \cite{cox1984analysis}.

In the presence of censoring, testing the hypothesis that the distribution of the lifetimes is Weibull is complicated by the fact that an incomplete sample is observed. \cite{balakrishnan2015empirical} suggests a way to perform the required goodness-of-fit tests by transforming the censored sample to a complete sample. Another approach is to modify the test statistics used in the full sample case to account for the presence of censoring. Although fewer in number, tests for the Weibull distribution in the presence of random censoring are available in the literature. For example, \cite{koziol1976cramer} and \cite{kim2017goodness} propose modified versions of the Cram\'er-von Mises test and the test proposed in \cite{liao1999new}, respectively, for use with censored data.

Throughout the paper we are primarily interested in the situation where censoring is present; the results relating to the full sample case are treated as special cases obtained when all lifetimes are observed. Before proceeding some notation is introduced. Let $X_1, \dots, X_n$ be independent and identically distributed (i.i.d.) lifetime variables with continuous distribution function $F$ and let $C_1, \dots, C_n$ be i.i.d. censoring variables with distribution function $H$, independent of $X_1, \dots, X_n$. We assume non-informative censoring throughout. Let
$$T_j=\mbox{min}(X_j,C_j) \ \ \ \ \ \ \mbox{and} \ \ \ \ \ \ \delta_j=
    \begin{cases}
      1, & \text{if}\ X_j\leq C_j, \\
      0, & \text{if}\ X_j > C_j.
    \end{cases}$$
Note that in the full sample case $T_j=X_j$ and $\delta_j=1$ for $j=1, \dots, n$.

Based on the observed pairs $(T_j, \delta_j), \ j=1, \dots n$ we wish to test the composite hypothesis
\begin{equation}\label{hypothesisW}
    H_0: X \sim Weibull(\lambda,\theta),
\end{equation}
for some unknown $\lambda>0$ and $\theta>0$. Here $X\sim Weibull(\lambda,\theta)$ refers to a Weibull distributed random variable with distribution function
    $$H(x) = 1-e^{-\left(x/\lambda\right)^\theta}, \ \ x>0.$$
This hypothesis is to be tested against general alternatives. We will make use of maximum likelihood estimation to estimate $\lambda$ and $\theta$. The log-likelihood of the Weibull distribution is
    $$\mathcal{L}(\theta,\lambda|X_1,\dots,X_n) = d\log(\theta)-d\theta\log(\lambda)+d(\theta-1)\log(X_j)-\lambda^{-\theta}\sum^n_{j=1} X_j^{\theta},$$
where $d=\sum^n_{j=1} \delta_j$. In the full sample case $d=n$. No closed form formulae for the maximum likelihood estimates $\hat{\lambda}$ and $\hat{\theta}$ exist, meaning that numerical optimisation techniques are required to arrive at parameter estimates.

Since the Weibull distribution has a shape parameter the first step for many goodness-of-fit tests for the Weibull distribution is to transform the data. If $X\sim Weibull(\lambda,\theta)$ then a frequently used transformation is $log(X)$, which results in a random variable that is type I extreme value distributed with parameters $\log(\lambda)$ and $1/\theta$. The resulting transformed random variable is part of a location scale family, which is a desirable result when performing goodness-of-fit testing. We therefore have that if $X \sim Weibull(\lambda,\theta)$, then $X^{(t)}=\theta(\log(X)-\log(\lambda))$ follows a standard type I extreme value distribution with distribution function
    $$G(x) = 1-\textrm{e}^{-\textrm{e}^{x}}, \ \ -\infty<x<\infty.$$
We denote a random variable with this distribution function by $EV(0,1)$. As a result the hypothesis in (\ref{hypothesisW}) holds, if, and only if, $X^{(t)}\sim EV(0,1)$. All of the test statistics considered make use of the transformed observed values
\begin{equation}\label{transform}
   Y_j = \hat{\theta}\left[\log (T_j)-\log(\hat{\lambda})\right], 
\end{equation}
with $\hat{\lambda}$ and $\hat{\theta}$ the maximum likelihood estimates of the Weibull distribution. Let 
$$X_j^{(t)} = \hat{\theta}\left[\log (X_j)-\log(\hat{\lambda})\right].$$
If $X_1,\dots,X_n$ are realised from a Weibull$(\lambda,\theta)$ distribution, then $X_1^{(t)},\dots,X_n^{(t)}$ will approximately be i.i.d. $EV(0,1)$ random variables, at least for large samples. As a result, the tests employed below are based on discrepancy measures between the calculated values of $Y_1,\dots,Y_n$ and the standard type I extreme value distribution. 
The order statistics of $Y_1, \dots ,Y_n$ are denoted by $Y_{(1)} < \cdots < Y_{(n)}$, while $\delta_{(j)}$ represents the indicator variable corresponding to $Y_{(j)}$.

The remainder of the paper is structured as follows. In Section \ref{Tests}, we propose two new classes of tests for the Weibull distribution for both the full sample and censored case. We also modify the test proposed in \cite{krit2014goodness} to accommodate random right censoring. In the presence of censoring the null distribution of all the test statistics considered depend on the unknown censoring distribution, we therefore propose a parametric bootstrap procedure in Section \ref{Bootstrap} in order to compute critical values. Section \ref{MonteCarlo} presents the results of a Monte Carlo study where the empirical powers of the newly proposed classes of tests as well as the newly modified test are compared to those of existing tests. The paper concludes in Section \ref{Conclude} with two practical applications; one concerning the survival times of patients diagnosed with a certain type of leukemia (no censoring is present in these data) and the other relates to observed leukemia remission times (in the presence of censoring). Some avenues for future research are also discussed.

\section{Proposed test statistics}
\label{Tests}

Our newly proposed classes of tests are based on the following theorem, which characterises the standard type I extreme value distribution.
\begin{theorem}\label{stein}
    Let $W$ be a random variable with absolutely continuous density and assume that $E\left[\textrm{e}^W\right]<\infty$. In this case $W\sim EV(0,1)$, if, and only if, 
    $$E\left[\left(it+1-\textrm{e}^W\right)\textrm{e}^{itW}\right]=0, \ \forall \ t \in R, $$
     with $i=\sqrt{-1}$.
\end{theorem}

\begin{proof}
    The theorem is an immediate consequence of Proposition 1.4 in \cite{diaconis2004use}, where we choose $f(x) = \textrm{e}^{itx}$ as the differentiable function specified in the mentioned paper. Straightforward calculations yield the desired result.
\end{proof}

Let $w(t)$ be a non-negative, symmetric weight function. From Theorem \ref{stein}, we have that 
\begin{align}\label{newtest}
       \eta &= \int_{-\infty}^{\infty} E\left[\left(it+1-\textrm{e}^Y\right)\textrm{e}^{itY}\right]w(t)\mathrm{d}t \notag \\ 
       &= \int_{-\infty}^{\infty} \int_{-\infty}^{\infty} \left(it+1-\textrm{e}^y\right)\textrm{e}^{ity} \mathrm{d}G(y) w(t)\mathrm{d}t 
\end{align}
equals $0$ if $Y\sim EV(0,1)$. Note that the inclusion of the weight function, $w$, above is require to ensure that $\eta$ is finite. Clearly $\eta$ will be unknown because $G$ is unknown. However, $G$ can be estimated by the Kaplan-Meier estimator, $G_n$, of the distribution function given by
\begin{equation}
    1-G_n(t)=
    \begin{cases}
      1, &  t \leq Y_{(1)} \\
      \prod_{j=1}^{k-1}\left(\frac{n-j}{n-j+1}\right)^{\delta_{(j)}}, & Y_{(k-1)}< t \leq Y_{(k)},  \ \ \ \ k=2,\dots, n. \nonumber\\
      \prod_{j=1}^n \left(\frac{n-j}{n-j+1}\right)^{\delta(j)}, & t>Y_{(n)}.
    \end{cases}
  \end{equation}
More details about this estimator can be found in \cite{kaplan1958nonparametric}, \cite{efron1967two} as well as \cite{breslow1974large}. In the full sample case this estimator reduces to the standard empirical distribution function, $G_n(x_{(j)})=j/n$.

Let $\Delta_j$ denote the size of the jump in $G_n(T_{(j)})$;
\begin{equation*}
    \Delta_j=G_n(T_{(j)})-\lim_{t \uparrow T_{(j)}}G_n(t),j=1,\ldots,n.
\end{equation*}
Simple calculable expressions for the $\Delta_j$'s are
\begin{eqnarray*}
\Delta_1 &=& \frac{\delta_{(1)}}{n}, \text{ } \Delta_n = \prod_{j=1}^{n-1}\left(\frac{n-j}{n-j+1}\right)^{\delta_{(j)}} \text{ and}\\
\Delta_j &=& \prod_{k=1}^{j-1} \left(\frac{n-k}{n-k+1}\right)^{\delta_{(k)}} - \prod_{k=1}^{j} \left(\frac{n-k}{n-k+1}\right)^{\delta_{(k)}}\\
&=& \frac{\delta_{(j)}}{n-j+1} \prod_{k=1}^{j-1} \left(\frac{n-k}{n-k+1}\right)^{\delta_{(k)}}, \ j=2,\dots,n-1.
\end{eqnarray*}
In the full sample case $\Delta_j=1/n, j=1,\ldots,n$.

Estimating $G$ by $G_n$ in (\ref{newtest}), we propose the test statistic
\begin{equation}
  S_{n,a} = n\int_{-\infty}^{\infty} \left|\sum_{j=1}^n \Delta_j \Big[it\textrm{e}^{itY_j}+(1-\textrm{e}^{Y_j})\textrm{e}^{itY_j}\Big]\right|^2 w_a(t) \mathrm{d}t,
\end{equation}
where $w_a(t)$ is a weight function containing a user-defined tuning parameter $a>0$. The null hypothesis in (\ref{hypothesisW}) is rejected for large values of $S_{n,a}$.

Straightforward algebra shows that, if $w_a(t)=\textrm{e}^{-at^2}$, then the test statistic simplifies to

\begin{align*}
    S_{n,a}^{(1)} &= n\sqrt{\frac{\pi}{a}} \sum_{j=1}^n\sum_{k=1}^n \Delta_j\Delta_k \textrm{e}^{-(Y_j-Y_k)^2/4a}\left\{-\frac{1}{4a^2}\left((Y_j-Y_k)^2-2a\right)\right. \\
    &+ \left. 2\left(1-\textrm{e}^{Y_j}\right)\left(\frac{1}{2a}\right)\left(Y_j-Y_k\right) + \left(1-\textrm{e}^{Y_j}\right)\left(1-\textrm{e}^{Y_k}\right) \right\},
\end{align*}
and if $w_a(t)=\textrm{e}^{-a|t|}$ the test statistic has the following easily calculable form
\begin{align*}
    S_{n,a}^{(2)} &= n\sum_{j=1}^n\sum_{k=1}^n \Delta_j\Delta_k \left\{ \frac{-4a\left(3(Y_j-Y_k)^2-a^2\right)}{\left((Y_j-Y_k)^2+a^2\right)^3}\right.\\
    &+\left. \frac{8a(Y_j-Y_k)\left(1-\textrm{e}^{Y_j}\right)}{\left((Y_j-Y_k)^2+a^2\right)^2} + \frac{2a\left(1-\textrm{e}^{Y_j}\right)\left(1-\textrm{e}^{Y_k}\right)}{(Y_j-Y_k)^2+a^2} \right\}.
\end{align*}
New goodness-of-fit tests containing a tuning parameter are often accompanied by a recommended value of this parameter; this choice is typically based on the finite sample power performance of the test. Another approach which may be used is to choose the value of the tuning parameter data-dependently; see, for example, \cite{AS2015}. In this paper, we opt to use the values recommended in the literature for tests containing a tuning parameter.

The weight functions specified above correspond to scaled Gaussian and Laplace kernels. These weight functions are popular choices found in the goodness-of-fit literature; see, for example, \cite{Meintanis2003}, \cite{AHM2017}, \cite{ebner2018normal}, \cite{betsch2019new} as well as \cite{HV2020}. The popularity of these weight functions is, at least in part, due to the fact that their inclusion typically results in simple calculable forms for $L^2$-type statistics which do not require numerical integration.

\textbf{Remarks on the asymptotic properties of $S_{n,a}$}

Although we do not derive the asymptotic results related to the proposed classes of test statistics we include some remarks in this regard. $S_{n,a}$ is a characteristic function based weighted $L^2$-type statistic. For a review of characteristic function based testing procedures, see \cite{meintanis2016review}. The asymptotic properties of this class of statistics are studied in detail in \cite{feuerverger1977empirical}, while more recent references include \cite{baringhaus1988consistent}, \cite{klar2005tests} as well as \cite{baringhaus2017limit}. A convenient setting for the derivation of the asymptotic properties of these tests is the separable Hilbert space of square integrable functions. These tests are usually consistent against a large class of fixed alternatives. Typically, the asymptotic null distribution of $S_{n,a}$ corresponds to that of $\int_{-\infty}^{\infty} \left|Z(t)\right|^2 w_a(t) \mathrm{d}t =: S_a$, where $Z(\cdot)$ is a zero-mean Gaussian process. The distribution of $S_a$ is the same as that of $\sum_{j=1}^{\infty} \lambda_jU_j$, where $U_j$ are i.i.d. chi-squared random variables with parameter $1$. However, the covariance matrix of $Z(\cdot)$ as well as the eigenvalues $\lambda_j, \ j\geq1$, depend on the unknown underlying distribution function, $F$. As a result, it is usually preferable to employ a bootstrap methodology in order to estimate critical values. Another advantage of this approach is that we are able to estimate these critical values in the case of finite samples and not just in the asymptotic case.

\section{Bootstrap algorithm}
\label{Bootstrap}
The null distribution of each of the test statistics considered depends on the unknown censoring distribution, even in the case of a simple hypothesis \cite{d1986goodness}. Since we will not assume any known form of the censoring distribution (e.g. the Kozoil-Green model), we propose the following parametric bootstrap algorithm to estimate the critical values of the tests.
\begin{enumerate}
	\item Based on the pairs $(T_j,\delta_j), \ j=1,\dots,n$ estimate $\theta$ and $\lambda$ by $\hat{\theta}$ and $\hat{\lambda}$, respectively, using maximum likelihood estimation.
	\item Transform $T_j$ to $Y_j$ using the transformation in (\ref{transform}) for $j=1,\dots,n$.
	\item Calculate the test statistic, say $W_n := W(Y_1,\ldots,Y_n;\delta_1,\ldots,\delta_n)$.
	\item Obtain a parametric bootstrap sample $X_1^*, \dots, X_n^*$ by sampling from a Weibull distribution with parameters $\hat{\theta}$ and $\hat{\lambda}$.
	\item Obtain a parametric bootstrap sample by sampling from the Kaplan-Meier estimate of the distribution of $C_j^{(t)} = \hat{\theta}\left[\log (C_j)-\log(\hat{\lambda})\right]$. 
	\item Set $$T_j^*=\mbox{min}(X_j^*,C_j^*) \ \mbox{and} \ \delta_j^*=
    \begin{cases}
      1, & \text{if}\ X_j^*\leq C_j^* \\
      0, & \text{if}\ X_j^* > C_j^*.
    \end{cases}$$ 
	\item Calculate $\hat{\theta}^*$ and $\hat{\lambda}^*$ based on $(T_j^*,\delta_j^*), \ j=1,\dots,n$.
	\item Obtain $Y_j^* = \hat{\theta}^*\left[\log(T_j^*)-\log(\hat{\lambda}^*)\right],\ j=1,\dots,n$.
	\item Based on the pairs $\left(Y_j^*,\delta_j^*\right), \ j=1,\dots,n$, calculate the value of the test statistic, say $W_n^* := W(Y_1^*,\ldots,Y_n^*;\delta_1^*,\ldots,\delta_n^*)$. 
	\item Repeat steps 4-9 B times to obtain $W_1^*, \dots, W_B^*$. Obtain the order statistics, $W_{(1)}^* \leq \dots \leq W_{(B)}^*$. The estimated critical value is then $\hat{c}_n(\alpha) = W_{\lfloor B(1-\alpha)\rfloor}^*$ where $\lfloor A\rfloor$ denotes the floor of $A$.
\end{enumerate}
The algorithm provided above is quite general and can easily be amended in order to test for any lifetime distribution in the presence of random censoring. In the absence of censoring there is no need to implement this algorithm; in this case, the critical values can be obtained via Monte Carlo simulation by sampling from any Weibull distribution and effecting the transforming discussed above.

\section{Numerical results}
\label{MonteCarlo}

In this section, we compare the power performances of the newly proposed tests to those of existing tests via a Monte Carlo simulation study. The existing tests used include the classical Kolmogorov-Smirnov ($KS_n$) and Cram\'er-von Mises ($CM_n$) tests. These tests have been modified for use with censored data, see \cite{koziol1976cramer}. The test introduced in \cite{liao1999new} is considered in the case of full samples. A modification making this test suitable for use with censored data is proposed in \cite{kim2017goodness}; we denote the test statistic by $LS_{n}$ in both the full sample and censored cases. The calculable forms of the test statistics mentioned above are
\begin{align*}
    KS_n &= \max\left[\max_{1\leq j \leq n} \left\{G_n(Y_{(j)})-\left(1-\textrm{e}^{-\textrm{e}^{Y_j}}\right)\right\},
    \right.\\
    &\left.\max_{1\leq j \leq n} \left\{\left(1-\textrm{e}^{-\textrm{e}^{Y_j}}\right)-G_n^-(Y_{(j)})\right\}\right],\\
    CM_n &= \frac{n}{3} + n\sum_{j=1}^{d+1} \left\{ G_n\left(X^{(t)}_{j-1}\right) \left(X^{(t)}_j-X^{(t)}_{j-1}\right)\right.\\
    &\left.\times\left[G_n\left(X^{(t)}_{j-1}\right)-\left(X^{(t)}_j+X^{(t)}_{j-1}\right)\right]\right\},\\
    LS_{n} &= \frac{1}{\sqrt{n}}\sum_{j=1}^n \left(\frac{\max\left[j/n-G_n(Y_j),G_n(Y_j)-(j-1)/n\right]}{\sqrt{G_n(Y_j)\left[1-G_n(Y_j)\right]}}\right).
\end{align*}

\cite{krit2014goodness} proposes a test for the Weibull distribution in the full sample case. This test compares the empirical Laplace transform of the random variables resulting from the transformation in (\ref{transform}) to the Laplace transform of an $EV(0,1)$ random variable; $\psi(t) = \Gamma(1-t)$ for $t<1$. Let $\psi_n$ be the empirical Laplace transform of the transformed observations, obtained using the Kaplan-Meier estimate of the distribution function;
\begin{equation}\label{laplace}
    \psi_n(t) = \int_{-\infty}^{\infty}\textrm{e}^{tx} \textrm{d}G_n(t) = \sum_{j=1}^n \Delta_j \textrm{e}^{-tY_j}.
\end{equation}
The resulting test statistic is
\begin{equation}\label{KritIntegral}
 KR^*_{m,a} = n\int_{I} \left[\psi_n(t)-\Gamma(1-t)\right]^2 w_a(t) \mathrm{d}t,
\end{equation}
where $w_a(t)= \textrm{e}^{at-\textrm{e}^{at}}$ is a weight function, $a$ user-specified tuning parameter and $I$ some interval. Based on numerical considerations, \cite{krit2014goodness} suggests that
$I=(-1,0]$ should be used. The quantity in (\ref{KritIntegral}) can be approximated by a Riemann sum;
\begin{equation*}
  KR_{m,a} = n\sum_{k=-m}^{-1}\left[\sum_{j=1}^n \Delta_j\textrm{e}^{-Y_jk/m}-\Gamma(1-k/m)\right]^2 \textrm{e}^{ak/m-\textrm{e}^{ak/m}},
\end{equation*}
where $m$ is the number of points at which the integrand is evaluated. In the numerical results shown below, we use  $a=-5$ and $m=100$ as recommended in \cite{krit2014goodness}. In the full sample case, $G_n$, in (\ref{laplace}), is taken to be the empirical distribution function. Upon setting $\Delta_j=1/n$ we obtain the test statistic in \cite{krit2014goodness}.

For each of the tests considered above, the null hypothesis in (\ref{hypothesisW}) is rejected for large values of the test statistics.

\subsection{Simulation setting}

In the numerical results presented below, we use a nominal significance level of 10\% throughout. Empirical powers are presented for sample sizes $n=50$ and $n=100$. The empirical powers for complete and censored samples are reported; censoring proportions of 10\% and 20\% are included. For each lifetime distribution considered, we report the powers obtained using three different censoring distributions; the exponential and uniform distributions as well as the Koziol-Green model, proposed in \cite{koziol1976cramer}. The alternative lifetime distributions considered are listed in Table \ref{table1}.

\begin{table}[!htbp]
	\centering\footnotesize
	\caption{Density functions of the alternative distributions.\label{table1}}
	\begin{tabular}{|c|c|c|}
	\hline
	Alternative & Density & Notation \\
	\hline
	Weibull & $\theta x^{\theta-1}\exp(-x^{\theta})$ & $W(\theta)$  \\
	Gamma & $\left(\Gamma(\theta)\right)^{-1}x^{\theta-1}\exp(-x)$ & $\Gamma(\theta)$\\
	Lognormal & $\left(\theta x \sqrt{2\pi}\right)^{-1} \exp\left(-{\log^2(x)}\left(2\theta^2\right)^{-1}\right)$ & $LN(\theta)$\\
	Chi square & $\left(2^{\theta/2}\Gamma(\theta/2)\right)^{-1}x^{\theta/2-1}\exp(-x/2)$ & $\chi^2(\theta)$  \\
	Beta & ${x^{\alpha-1}(1-x)^{\theta-1}\Gamma(\alpha+\theta)}\left(\Gamma(\alpha)\Gamma(\theta)\right)^{-1}$ & $\beta(\alpha,\theta)$\\
	Lindley & $\frac{\theta^2}{\theta+1}(1+x)\exp(-\theta x)$ & $Lind(\theta)$  \\
	\hline
	\end{tabular}\normalsize
\end{table} 

The obtained empirical powers are presented in Tables \ref{table2} to \ref{table7}. These tables report the percentages of 50\,000 independent Monte Carlo samples that lead to the rejection of the null hypothesis, rounded to the nearest integer. For ease of comparison, the highest power in each line is printed in bold. Tables \ref{table2} and \ref{table3} contain the results relating to full samples. For each test considered, Tables \ref{table4} to \ref{table7} show three empirical powers against each lifetime distribution, corresponding to the three different censoring distributions used. In each case, the results for the exponential, uniform and Koziol-Green models are shown in the first, second and third lines, respectively.

In order to reduce the computational cost associated with the numerical powers a warp-speed bootstrap procedure, see \cite{GPW:2013}, is employed. This methodology has been employed by a number of authors in the literature to compare Monte Carlo performances; see, for example, \cite{meintanis2018testing}, \cite{allison2019} as well as \cite{MV2020}. The bootstrap algorithm in Section \ref{Bootstrap} is implemented to calculate the critical values used to obtain the results in Tables \ref{table4}-\ref{table7}.

In the discussion below, including the tables, the subscript $n$ is suppressed. For $S_{a}^{(1)}$ and $S_{a}^{(2)}$, we include numerical powers in the cases where $a$ is set to $1, \ 5$ and $10$. All calculations are performed in \textsf{R} \cite{CRAN}. The $\emph{LindleyR}$ package is used to generate samples from censored distributions, see \cite{Lind}. Parameter estimation is performed using the $\emph{parmsurvfit}$ package, see \cite{parmsurvfit}, while 
the tables are produced using the $\emph{Stargazer}$ package, see \cite{Starg}.

\subsection{Simulation results}
First, we consider the results associated with the full sample case, given in Tables \ref{table2} and \ref{table3}. All of the tests considered attain the nominal size for both sample sizes used. The tests associated with the highest powers are $S_{5}^{(1)}$ and $S_{5}^{(2)}$, although $KR_{m,a}$ also performs well, especially for smaller samples. In general, the newly proposed tests outperform the other tests for the majority of the alternatives considered. When analysing complete samples, we recommend using $S_{5}^{(1)}$ or $S_{5}^{(2)}$.

We now turn our attention to the powers achieved in the presence of censoring. The size of the tests are maintained closely for all sample sizes for censoring proportions of 10\% and 20\%, with the single exception of $KR_{m,a}$ in the case of 20\% censoring. As expected, the powers generally increase with the sample size and decrease marginally as the censoring proportion increases.

For all combinations of sample sizes and censoring proportions $KR_{m,a}$ and $S_1^{(1)}$ generally tend to provide the highest powers. However, it should be noted that $KR_{m,a}$ achieves very low power against certain alternatives. For instance, when $n=100$ and the censoring proportion is $10\%$, the powers associated with $KR_{m,a}$ and $S_1^{(1)}$ against the $\beta(0.5,1)$ distribution are $0$ and $98$, respectively. Other examples where $KR_{m,a}$ provides low powers include the $Lind$ distributions.

When compiling the numerical results, we also considered a wider range of values for the tuning parameter, $a$, than those reported in the table. Although some power variation is evident when varying $a$, the powers achieved by the newly proposed classes of tests are not particularly sensitive to the choice of the tuning parameter $a$.

Some remarks are in order when considering the impact of the censoring distribution on the attained powers. First, this impact becomes more pronounced as the censoring proportion increases. Second, the powers associated with the Koziol-Green censoring model are noticeably lower than those associated with the uniform or exponential distributions in the majority of cases. Research as to the reason for this discrepancy is currently underway. Informally, we believe that the discrepancy is due to a higher degree of similarity between the censoring distribution and the hypothesised lifetime distribution than is the case for the other censoring distributions considered.

\section{Practical application and conclusion}
\label{Conclude}
In this section, we use the tests discussed in Section \ref{MonteCarlo} to test the hypothesis in (\ref{hypothesisW}) based on two real-world data sets. The first data set, reported in Table \ref{table10}, contains the survival times, in days, of 43 Leukemia patients. For a discussion of the original data set see \cite{kotze1983encyclopedia} as well as \cite{allison2017apples}. This data set is not subject to censoring, i.e. all lifetimes are observed. The second data set contains the initial remission times of leukemia patients, in days; for more details see, \cite{lee2003statistical}, this data set can be found in Table \ref{table8}. This data set contains censored observations, indicated using an asterisk. The original data were segmented into three treatment groups. However, \cite{lee2003statistical} showed that the data do not display significant differences among the various treatments. As a result we treat the data as i.i.d. realisations from a single, censored, lifetime distribution. All reported p-values are estimated using 1 million bootstrap replications; these results are displayed in Tables \ref{table11} and \ref{table9}, respectively.

\begin{table}[!htbp] \centering 
  \caption{Survival times after leukemia diagnosis, in days.} 
  \label{table10} 
\begin{tabular}{@{\extracolsep{5pt}} c} 
\hline
$7, 47, 58, 74, 177, 232, 273, 285, 317, 429, 440, 445,455, 468, 495, 497, 532,$\\
$ 571, 579, 581, 650, 702, 715, 779,881, 900, 930, 968, 1077, 1109, 1314,$\\
$ 1334, 1367, 1534, 1712, 1784,1877, 1886, 2045, 2056, 2260, 2429, 2509$\\
\hline
\end{tabular} 
\end{table}

\begin{table}[!htbp] \centering 
  \caption{$p$-values associated with the various tests used in the full sample case.} 
  \label{table11} 
\begin{tabular}{ccccccccccc} 
\hline
$Test$ & $KS$ & $CM$ & $LS$ & $KR$ & $S_{1}^{(1)}$ & $S_{5}^{(1)}$ & $S_{10}^{(1)}$ & $S_{1}^{(2)}$ & $S_{5}^{(2)}$ & $S_{10}^{(2)}$ \\
\hline 
$p$-value & $0.93$ & $0.97$ & $0.89$ & $0.58$ & $0.91$ & $0.53$ & $0.48$ & $0.87$ & $0.62$ & $0.49$ \\ 
\hline 
\end{tabular}
\end{table}

From the results of the practical example, in Table \ref{table11}, it is clear that none of the tests reject the null hypothesis that the survival times after a leukemia diagnosis are Weibull distributed at the 5\% or 10\% levels of significance. As a result, we conclude that the Weibull distribution is an appropriate model for these data.

\begin{table}[!htbp] \centering 
  \caption{Initial remission times of leukemia patients, in days.} 
  \label{table8} 
\begin{tabular}{@{\extracolsep{5pt}} c} 
\hline
$4,5,8,8,9,10,10,10,10,10,11,12,12,12^*,13,14,20,20^*,23,23,25,25,25,28,28,28,$\\
$28,29,31,31,31,32,37,40,41,41,48,48,57,62,70,74,75,89,99,100,103,124,139,143,$\\
$159^*,161^*,162,169,190^*,195,196^*,197^*,199^*,205^*,217^*,219^*,220,245^*,258^*,269^*$\\
\hline
\end{tabular} 
\end{table} 

\begin{table}[!htbp] \centering 
  \caption{$p$-values associated with the various tests used in the censored case.} 
  \label{table9} 
\begin{tabular}{ccccccccccc} 
\hline
$Test$ & $KS$ & $CM$ & $LS$ & $KR$ & $S_{1}^{(1)}$ & $S_{5}^{(1)}$ & $S_{10}^{(1)}$ & $S_{1}^{(2)}$ & $S_{5}^{(2)}$ & $S_{10}^{(2)}$ \\
\hline 
$p$-value & $0.03$ & $0.04$ & $0.08$ & $0.28$ & $0.05$ & $0.08$ & $0.11$ & $0.11$ & $0.098$ & $0.13$ \\ 
\hline 
\end{tabular}
\end{table}

The results associated with the initial remission times, in Table \ref{table9}, indicate that $KS$ rejects the hypothesis in (\ref{hypothesisW}) at a 5\% significance level. $CM$ and $S_1^{(1)}$ also provide some evidence against the null hypothesis with $p$-values of 4\% and 5\% respectively. However, none of the remaining 7 tests considered result in a rejection of the null hypothesis at the 5\% level. We conclude that the Weibull distribution is likely to be an appropriate model for the observed times. The data set under consideration was also analysed in \cite{bothma2020kaplan}, where the null hypothesis of exponentiality of the remission time was strongly rejected. The mentioned paper recommended that a more flexible distribution be used when modelling these data. The results above indicate that the additional flexibility of the Weibull (compared to the exponential) distribution indeed ensures that the Weibull distribution is a more appropriate model than the exponential for the initial remission times considered.


A number of interesting numerical phenomena are evident when considering the powers of the various tests. It is clear that the achieved powers, and, therefore, the null distribution of the test statistic, is influenced by the shape of the censoring distribution. The effect of the censoring distribution on the critical values of the tests seem not to have been investigated in the literature to date. Some authors perform goodness-of-fit testing by enforcing a parametric assumption on the censoring distribution, see, for example, \cite{kim2017goodness}. An additional consideration that seems to have been neglected in the literature is the effect on the null distribution of the test statistic, and hence the power of the test, of a specific assumption made in the Kaplan-Meier estimate of the distribution function. Some authors, in order to ensure that $G_n$ satisfies the requirements of a distribution, defines $G_n(x_{(n)})=1$. We are currently investigating these open questions.


%
%

\bibliographystyle{unsrt}  
\bibliography{Article2}  

\begin{thebibliography}{10}

\bibitem{krit2014goodness}
Meryam Krit.
\newblock Goodness-of-fit tests for the {W}eibull distribution based on the
  laplace transform.
\newblock {\em Journal de la Soci{\'e}t{\'e} Fran{\c{c}}aise de Statistique},
  155(3):135--151, 2014.

\bibitem{kalbfleisch2011statistical}
John~D Kalbfleisch and Ross~L Prentice.
\newblock {\em The statistical analysis of failure time data}, volume 360.
\newblock John Wiley \& Sons, 2011.

\bibitem{jiang2011study}
Renyan Jiang and DNP Murthy.
\newblock A study of {W}eibull shape parameter: {P}roperties and significance.
\newblock {\em Reliability Engineering \& System Safety}, 96(12):1619--1626,
  2011.

\bibitem{mann1973men}
Nancy~R Mann, Ernest~M Scneuer, and Kanneth~W Fertig.
\newblock A new goodness-of-fit test for the two-parameter {W}eibull or
  extreme-value distribution with unknown parameters.
\newblock {\em Communications in Statistics-Theory and Methods}, 2(5):383--400,
  1973.

\bibitem{tiku1981testing}
ML~Tiku and M~Singh.
\newblock Testing the two parameter {W}eibull distribution.
\newblock {\em Communications in Statistics-Theory and Methods},
  10(9):907--918, 1981.

\bibitem{liao1999new}
Min Liao and Toshiyuki Shimokawa.
\newblock A new goodness-of-fit test for type-i extreme-value and 2-parameter
  {W}eibull distributions with estimated parameters.
\newblock {\em Optimization}, 64(1):23--48, 1999.

\bibitem{cabana2005using}
Alejandra Caba{\~n}a and Adolfo~J Quiroz.
\newblock Using the empirical moment generating function in testing for the
  {W}eibull and the type {I} extreme value distributions.
\newblock {\em Test}, 14(2):417--431, 2005.

\bibitem{cox1984analysis}
David~Roxbee Cox and David Oakes.
\newblock {\em Analysis of {S}urvival {D}ata}, volume~21.
\newblock CRC Press, 1984.

\bibitem{balakrishnan2015empirical}
Narayanaswamy Balakrishnan, Ekaterina Chimitova, and M~Vedernikova.
\newblock An empirical analysis of some nonparametric goodness-of-fit tests for
  censored data.
\newblock {\em Communications in Statistics-Simulation and Computation},
  44(4):1101--1115, 2015.

\bibitem{koziol1976cramer}
James~A Koziol and Sylvan~B Green.
\newblock A {C}ram{\'e}r-von {M}ises statistic for randomly censored data.
\newblock {\em Biometrika}, 63(3):465--474, 1976.

\bibitem{kim2017goodness}
Namhyun Kim.
\newblock Goodness-of-fit tests for randomly censored {W}eibull distributions
  with estimated parameters.
\newblock {\em Communications for Statistical Applications and Methods},
  24(5):519--531, 2017.

\bibitem{diaconis2004use}
Persi Diaconis, Charles Stein, Susan Holmes, and Gesine Reinert.
\newblock Use of exchangeable pairs in the analysis of simulations.
\newblock In {\em Stein's Method}, pages 1--25. Institute of Mathematical
  Statistics, 2004.

\bibitem{kaplan1958nonparametric}
Edward~L Kaplan and Paul Meier.
\newblock Nonparametric estimation from incomplete observations.
\newblock {\em Journal of the American Statistical Association},
  53(282):457--481, 1958.

\bibitem{efron1967two}
Bradley Efron.
\newblock The two sample problem with censored data.
\newblock In {\em Proceedings of the Fifth Berkeley Symposium on Mathematical
  Statistics and Probability}, volume~4, pages 831--853, 1967.

\bibitem{breslow1974large}
Norman Breslow and John Crowley.
\newblock A large sample study of the life table and product limit estimates
  under random censorship.
\newblock {\em The Annals of statistics}, pages 437--453, 1974.

\bibitem{AS2015}
J.S. Allison and L.~Santana.
\newblock On a data-dependent choice of the tuning parameter appearing in
  certain goodness-of-fit tests.
\newblock {\em Journal of Statistical Computation and Simulation},
  85(16):3276--3288, 2015.

\bibitem{Meintanis2003}
S.~G. Meintanis and G.~Iliopoulos.
\newblock Tests of fit for the {R}ayleigh distribution based on the empirical
  {L}aplace transform.
\newblock {\em Annals of the Institute of Statistical Mathematics},
  55(1):137--151, 2003.

\bibitem{AHM2017}
J.~S. Allison, M.~Huskova, and S.~G. Meintanis.
\newblock Testing the adequacy of semiparametric transformation models.
\newblock {\em Test}, 27:1--25, 2017.

\bibitem{ebner2018normal}
Steffen Betsch and Bruno Ebner.
\newblock Testing normality via a distributional fixed point property in the
  {S}tein characterization.
\newblock {\em TEST}, pages 1--34, 2018.

\bibitem{betsch2019new}
Steffen Betsch and Bruno Ebner.
\newblock A new characterization of the gamma distribution and associated
  goodness-of-fit tests.
\newblock {\em Metrika}, 82(7):779--806, 2019.

\bibitem{HV2020}
N.~Henze and I.~J.~H. Visagie.
\newblock Testing for normality in any dimension based on a partial
  differential equation involving the moment generating function.
\newblock {\em Ann. Inst. Stat. Math.}, 72:1109--1136, 2020.

\bibitem{meintanis2016review}
Simos~G Meintanis.
\newblock A review of testing procedures based on the empirical characteristic
  function.
\newblock {\em South African Statistical Journal}, 50(1):1--14, 2016.

\bibitem{feuerverger1977empirical}
Andrey Feuerverger and Roman~A. Mureika.
\newblock The empirical characteristic function and its applications.
\newblock {\em The annals of Statistics}, 5(1):88--97, 1977.

\bibitem{baringhaus1988consistent}
Ludwig Baringhaus and Norbert Henze.
\newblock A consistent test for multivariate normality based on the empirical
  characteristic function.
\newblock {\em Metrika}, 35(1):339--348, 1988.

\bibitem{klar2005tests}
Bernhard Klar and Simos~G Meintanis.
\newblock Tests for normal mixtures based on the empirical characteristic
  function.
\newblock {\em Computational statistics \& data analysis}, 49(1):227--242,
  2005.

\bibitem{baringhaus2017limit}
L~Baringhaus, B~Ebner, and N~Henze.
\newblock The limit distribution of weighted l 2-goodness-of-fit statistics
  under fixed alternatives, with applications.
\newblock {\em Annals of the Institute of Statistical Mathematics},
  69(5):969--995, 2017.

\bibitem{d1986goodness}
Ralph~B D'Agostino and Michael~A Stephens.
\newblock {\em Goodness-of-fit {T}echniques}, volume~68.
\newblock CRC press, 1986.

\bibitem{GPW:2013}
R.~Giacomini, D.~N. Politis, and H.~White.
\newblock A warp-speed method for conducting {M}onte {C}arlo experiments
  involving bootstrap estimators.
\newblock {\em Econometric Theory}, 29(3):567--589, 2013.

\bibitem{meintanis2018testing}
Simos~G Meintanis, Joseph Ngatchou-Wandji, and James~S. Allison.
\newblock Testing for serial independence in vector autoregressive models.
\newblock {\em Statistical Papers}, 59(4):1379--1410, 2018.

\bibitem{allison2019}
J.~S. Allison, S.~Betsch, B.~Ebner, and I.~J.~H. Visagie.
\newblock New weighted $l^2$-type tests for the inverse {G}aussian
  distribution.
\newblock {\em arXiv preprint arXiv:1910.14119}, 2019.

\bibitem{MV2020}
P.~A. Mijburgh and I.~J.~H. Visagie.
\newblock An overview of goodness-of-fit tests for the poisson distribution.
\newblock {\em South African Statistical Journal}, 54(2):207--230, 2020.

\bibitem{CRAN}
{R Core Team}.
\newblock {\em R: {A} Language and Environment for Statistical Computing}.
\newblock R Foundation for Statistical Computing, Vienna, Austria, 2019.

\bibitem{Lind}
Josmar Mazucheli, Larissa~B. Fernandes, and Ricardo~P. {de Oliveira}.
\newblock {\em Lindley{R}: The {L}indley Distribution and Its Modifications},
  2016.
\newblock R package version 1.1.0.

\bibitem{parmsurvfit}
Ashley Jacobson, Victor Wilson, and Shannon Pileggi.
\newblock {\em parmsurvfit: Parametric Models for Survival Data}, 2018.
\newblock R package version 0.1.0.

\bibitem{Starg}
M.~Hlavac.
\newblock {\em stargazer: Well-formatted regression and summary statistics
  tables}, 2018.

\bibitem{kotze1983encyclopedia}
S~Kotze and NL~Johnson.
\newblock {\em Encyclopedia of statistical sciences}, volume~3.
\newblock Wiley, New York, 1983.

\bibitem{allison2017apples}
J.~S. Allison, L~Santana, N~Smit, and I.~J.~H. Visagie.
\newblock An "apples-to-apples" comparison of various tests for exponentiality.
\newblock {\em Computational Statistics}, 32(4):1241--1283, 2017.

\bibitem{lee2003statistical}
Elisa~T Lee and John Wang.
\newblock {\em Statistical {M}ethods for {S}urvival {D}ata {A}nalysis}, volume
  476.
\newblock John Wiley \& Sons, 2003.

\bibitem{bothma2020kaplan}
E~Bothma, JS~Allison, M~Cockeran, and IJH Visagie.
\newblock Kaplan-meier based tests for exponentiality in the presence of
  censoring.
\newblock {\em arXiv preprint arXiv:2011.04519}, 2020.

\end{thebibliography}

\begin{table}[!htbp] \centering 
  \caption{Estimated powers for the full sample case where n=50} 
  \label{table2}
\begin{tabular}{@{\extracolsep{1pt}} ccccccccccc} 
\hline
$F$ & $KS$ & $CM$ & $LS$ & $KR$ & $S_{1}^{(1)}$ & $S_{5}^{(1)}$ & $S_{10}^{(1)}$ & $S_{1}^{(2)}$ & $S_{5}^{(2)}$ & $S_{10}^{(2)}$ \\
\hline 
$W(0.5)$ & $10$ & $10$ & $10$ & $10$ & $10$ & $10$ & $10$ & $10$ & $10$ & $10$ \\ 
$W(1.5)$ & $10$ & $10$ & $10$ & $10$ & $10$ & $10$ & $10$ & $10$ & $10$ & $10$ \\ 
$W(2)$ & $10$ & $10$ & $10$ & $10$ & $10$ & $10$ & $10$ & $10$ & $10$ & $10$ \\ 
$\Gamma(2)$ & $13$ & $15$ & $12$ & $\textbf{17}$ & $16$ & $16$ & $15$ & $14$ & $16$ & $15$ \\ 
$\Gamma(3)$ & $18$ & $20$ & $16$ & $\textbf{25}$ & $24$ & $24$ & $24$ & $18$ & $\textbf{25}$ & $24$ \\ 
$LN(0.5)$ & $50$ & $61$ & $51$ & $70$ & $73$ & $76$ & $75$ & $50$ & $\textbf{77}$ & $75$ \\ 
$LN(1)$ & $50$ & $61$ & $50$ & $70$ & $72$ & $\textbf{76}$ & $75$ & $50$ & $\textbf{76}$ & $75$ \\ 
$\chi^2(8)$ & $20$ & $24$ & $19$ & $\textbf{31}$ & $29$ & $\textbf{31}$ & $30$ & $20$ & $\textbf{31}$ & $30$ \\ 
$\chi^2(10)$ & $23$ & $28$ & $22$ & $35$ & $34$ & $\textbf{36}$ & $35$ & $23$ & $\textbf{36}$ & $35$ \\ 
$\beta(1,1)$ & $69$ & $81$ & $76$ & $\textbf{92}$ & $85$ & $88$ & $86$ & $79$ & $89$ & $86$ \\ 
$\beta(0.5,1)$ & $69$ & $81$ & $76$ & $\textbf{92}$ & $85$ & $88$ & $86$ & $79$ & $89$ & $86$ \\ 
$Lind(0.5)$ & $11$ & $11$ & $13$ & $12$ & $11$ & $14$ & $\textbf{15}$ & $10$ & $14$ & $\textbf{15}$ \\ 
$Lind(2)$ & $10$ & $10$ & $11$ & $11$ & $10$ & $11$ & $\textbf{12}$ & $10$ & $11$ & $\textbf{12}$ \\ 
\hline \\[-1.8ex] 
\end{tabular} 
\end{table}

\begin{table}[!htbp] \centering 
  \caption{Estimated powers for the full sample case where n=100} 
  \label{table3} 
\begin{tabular}{@{\extracolsep{1pt}} ccccccccccc} 
\hline
$F$ & $KS$ & $CM$ & $LS$ & $KR$ & $S_{1}^{(1)}$ & $S_{5}^{(1)}$ & $S_{10}^{(1)}$ & $S_{1}^{(2)}$ & $S_{5}^{(2)}$ & $S_{10}^{(2)}$ \\
\hline 
$W(0.5)$ & $10$ & $10$ & $10$ & $10$ & $10$ & $10$ & $10$ & $10$ & $10$ & $10$ \\ 
$W(1.5)$ & $10$ & $10$ & $10$ & $10$ & $10$ & $10$ & $10$ & $10$ & $10$ & $10$ \\ 
$W(2)$ & $10$ & $10$ & $10$ & $10$ & $10$ & $10$ & $10$ & $10$ & $10$ & $10$ \\ 
$\Gamma(2)$ & $18$ & $20$ & $17$ & $26$ & $24$ & $\textbf{27}$ & $26$ & $17$ & $\textbf{27}$ & $26$ \\ 
$\Gamma(3)$ & $26$ & $30$ & $25$ & $41$ & $39$ & $\textbf{44}$ & $\textbf{44}$ & $24$ & $43$ & $\textbf{44}$ \\ 
$LN(0.5)$ & $77$ & $88$ & $79$ & $93$ & $95$ & $\textbf{97}$ & $\textbf{97}$ & $78$ & $\textbf{97}$ & $\textbf{97}$ \\ 
$LN(1)$ & $78$ & $88$ & $80$ & $94$ & $95$ & $\textbf{97}$ & $\textbf{97}$ & $78$ & $\textbf{97}$ & $\textbf{97}$ \\ 
$\chi^2(8)$ & $32$ & $39$ & $32$ & $51$ & $49$ & $\textbf{56}$ & $55$ & $30$ & $55$ & $55$ \\ 
$\chi^2(10)$ & $36$ & $44$ & $36$ & $58$ & $56$ & $\textbf{64}$ & $63$ & $34$ & $63$ & $63$ \\ 
$\beta(1,1)$ & $94$ & $98$ & $96$ & $\textbf{100}$ & $99$ & $99$ & $99$ & $99$ & $99$ & $99$ \\ 
$\beta(0.5,1)$ & $94$ & $98$ & $96$ & $\textbf{100}$ & $99$ & $99$ & $99$ & $99$ & $\textbf{100}$ & $99$ \\
$Lind(0.5)$ & $11$ & $12$ & $14$ & $13$ & $12$ & $\textbf{16}$ & $\textbf{16}$ & $10$ & $15$ & $\textbf{16}$ \\ 
$Lind(2)$ & $10$ & $10$ & $11$ & $11$ & $11$ & $\textbf{12}$ & $\textbf{12}$ & $10$ & $\textbf{12}$ & $\textbf{12}$ \\ 
\hline \\[-1.8ex] 
\end{tabular} 
\end{table}

\begin{table}[!htbp] \centering 
  \caption{Estimated powers for $10\%$ censoring for a sample size of n=50 with three different censoring distributions.} 
  \label{table4} 
\begin{tabular}{@{\extracolsep{1pt}} ccccccccccc} 
\hline
$F$ & $KS$ & $CM$ & $LS$ & $KR$ & $S_{1}^{(1)}$ & $S_{5}^{(1)}$ & $S_{10}^{(1)}$ & $S_{1}^{(2)}$ & $S_{5}^{(2)}$ & $S_{10}^{(2)}$ \\
\hline 
\multirow{3}[2]{*}{$W(0.5)$} & $9$ & $9$ & $8$ & $7$ & $8$ & $8$ & $8$ & $8$ & $8$ & $8$ \\ 
 & $9$ & $9$ & $8$ & $4$ & $9$ & $9$ & $8$ & $8$ & $9$ & $8$ \\ 
 & $10$ & $11$ & $11$ & $7$ & $10$ & $11$ & $10$ & $10$ & $11$ & $10$ \\
\hline 
\multirow{3}[2]{*}{$W(1.5)$} & $9$ & $10$ & $9$ & $9$ & $9$ & $9$ & $8$ & $9$ & $9$ & $8$ \\ 
 & $10$ & $10$ & $9$ & $9$ & $9$ & $9$ & $8$ & $9$ & $9$ & $8$ \\ 
 & $10$ & $10$ & $9$ & $9$ & $8$ & $10$ & $10$ & $9$ & $10$ & $10$ \\
\hline 
\multirow{3}[2]{*}{$W(2)$} & $9$ & $10$ & $10$ & $8$ & $9$ & $9$ & $8$ & $10$ & $9$ & $8$ \\ 
 & $10$ & $10$ & $10$ & $9$ & $9$ & $9$ & $8$ & $10$ & $9$ & $8$ \\ 
 & $10$ & $9$ & $9$ & $9$ & $8$ & $10$ & $10$ & $9$ & $10$ & $10$ \\ 
\hline 
\multirow{3}[2]{*}{$\Gamma(2)$} & $13$ & $14$ & $10$ & $\textbf{21}$ & $15$ & $12$ & $10$ & $13$ & $13$ & $10$ \\ 
 & $12$ & $14$ & $10$ & $\textbf{21}$ & $15$ & $13$ & $11$ & $13$ & $13$ & $11$ \\ 
 & $10$ & $11$ & $8$ & $\textbf{15}$ & $10$ & $11$ & $11$ & $10$ & $11$ & $11$ \\
\hline 
\multirow{3}[2]{*}{$\Gamma(3)$} & $17$ & $19$ & $14$ & $\textbf{29}$ & $22$ & $19$ & $15$ & $17$ & $20$ & $13$ \\ 
 & $16$ & $18$ & $13$ & $\textbf{29}$ & $21$ & $18$ & $14$ & $16$ & $19$ & $12$ \\ 
 & $12$ & $12$ & $9$ & $\textbf{19}$ & $12$ & $12$ & $11$ & $12$ & $12$ & $11$ \\ 
\hline 
\multirow{3}[2]{*}{$LN(0.5)$} & $45$ & $56$ & $42$ & $65$ & $\textbf{67}$ & $64$ & $56$ & $46$ & $65$ & $44$ \\ 
 & $45$ & $55$ & $40$ & $64$ & $\textbf{66}$ & $62$ & $52$ & $45$ & $63$ & $40$ \\ 
 & $23$ & $27$ & $18$ & $\textbf{41}$ & $33$ & $23$ & $16$ & $23$ & $23$ & $15$ \\ 
\hline 
\multirow{3}[2]{*}{$LN(1)$} & $43$ & $54$ & $32$ & $56$ & $\textbf{61}$ & $33$ & $23$ & $43$ & $33$ & $22$ \\ 
 & $42$ & $53$ & $26$ & $47$ & $\textbf{59}$ & $32$ & $30$ & $41$ & $32$ & $30$ \\ 
 & $22$ & $27$ & $18$ & $\textbf{41}$ & $32$ & $22$ & $16$ & $23$ & $23$ & $15$ \\  
\hline 
\multirow{3}[2]{*}{$\chi^2(8)$}  & $19$ & $22$ & $16$ & $\textbf{34}$ & $27$ & $24$ & $20$ & $19$ & $25$ & $17$ \\ 
 & $19$ & $23$ & $16$ & $\textbf{34}$ & $27$ & $24$ & $19$ & $19$ & $25$ & $16$ \\ 
 & $12$ & $13$ & $9$ & $\textbf{21}$ & $14$ & $12$ & $11$ & $12$ & $12$ & $11$ \\ 
\hline 
\multirow{3}[2]{*}{$\chi^2(10)$}  & $22$ & $26$ & $19$ & $\textbf{37}$ & $31$ & $29$ & $25$ & $21$ & $30$ & $20$ \\ 
 & $21$ & $25$ & $19$ & $\textbf{38}$ & $31$ & $28$ & $23$ & $21$ & $29$ & $19$ \\ 
 & $13$ & $15$ & $10$ & $\textbf{23}$ & $15$ & $12$ & $11$ & $13$ & $12$ & $11$ \\ 
\hline 
\multirow{3}[2]{*}{$\beta(1,1)$} & $63$ & $75$ & $71$ & $1$ & $80$ & $80$ & $75$ & $73$ & $\textbf{81}$ & $69$ \\ 
 & $63$ & $75$ & $70$ & $1$ & $79$ & $79$ & $74$ & $73$ & $\textbf{81}$ & $68$ \\ 
 & $\textbf{10}$ & $\textbf{10}$ & $\textbf{10}$ & $5$ & $8$ & $\textbf{10}$ & $\textbf{10}$ & $9$ & $\textbf{10}$ & $\textbf{10}$ \\  
\hline 
\multirow{3}[2]{*}{$\beta(0.5,1)$} & $61$ & $72$ & $67$ & $0$ & $\textbf{75}$ & $72$ & $62$ & $70$ & $74$ & $46$ \\ 
 & $60$ & $71$ & $67$ & $0$ & $\textbf{74}$ & $70$ & $58$ & $69$ & $71$ & $40$ \\ 
 & $16$ & $\textbf{18}$ & $9$ & $17$ & $15$ & $11$ & $11$ & $15$ & $11$ & $11$ \\
\hline 
\multirow{3}[2]{*}{$Lind(0.5)$}  & $10$ & $11$ & $\textbf{12}$ & $8$ & $10$ & $11$ & $10$ & $9$ & $11$ & $9$ \\ 
 & $10$ & $11$ & $\textbf{12}$ & $8$ & $10$ & $\textbf{12}$ & $11$ & $9$ & $11$ & $10$ \\
 & $10$ & $\textbf{11}$ & $\textbf{11}$ & $4$ & $9$ & $10$ & $9$ & $9$ & $9$ & $9$ \\
\hline 
\multirow{3}[2]{*}{$Lind(2)$}  & $\textbf{10}$ & $\textbf{10}$ & $\textbf{10}$ & $7$ & $8$ & $\textbf{10}$ & $9$ & $9$ & $\textbf{10}$ & $9$ \\ 
 & $10$ & $10$ & $10$ & $6$ & $9$ & $\textbf{11}$ & $\textbf{11}$ & $9$ & $\textbf{11}$ & $\textbf{11}$ \\ 
& $\textbf{10}$ & $\textbf{10}$ & $9$ & $8$ & $8$ & $\textbf{10}$ & $\textbf{10}$ & $9$ & $\textbf{10}$ & $\textbf{10}$ \\
\hline \\[-1.8ex] 
\end{tabular} 
\end{table}

\begin{table}[!htbp] \centering 
  \caption{Estimated powers for $20\%$ censoring for a sample size of n=50 with three different censoring distributions.} 
  \label{table5} 
\begin{tabular}{@{\extracolsep{1pt}} ccccccccccc} 
\hline
$F$ & $KS$ & $CM$ & $LS$ & $KR$ & $S_{1}^{(1)}$ & $S_{5}^{(1)}$ & $S_{10}^{(1)}$ & $S_{1}^{(2)}$ & $S_{5}^{(2)}$ & $S_{10}^{(2)}$ \\
\hline 
\multirow{3}[2]{*}{$W(0.5)$} & $7$ & $6$ & $7$ & $5$ & $8$ & $7$ & $7$ & $7$ & $7$ & $7$ \\ 
 & $8$ & $6$ & $7$ & $2$ & $8$ & $8$ & $7$ & $8$ & $7$ & $7$ \\ 
 & $9$ & $8$ & $9$ & $8$ & $8$ & $7$ & $7$ & $8$ & $7$ & $7$ \\
\hline 
\multirow{3}[2]{*}{$W(1.5)$} & $9$ & $9$ & $9$ & $8$ & $8$ & $10$ & $10$ & $8$ & $10$ & $10$ \\ 
 & $9$ & $9$ & $8$ & $8$ & $8$ & $10$ & $10$ & $8$ & $10$ & $10$ \\ 
 & $9$ & $8$ & $8$ & $8$ & $7$ & $7$ & $7$ & $8$ & $7$ & $7$ \\
\hline 
\multirow{3}[2]{*}{$W(2)$} & $9$ & $9$ & $9$ & $7$ & $8$ & $9$ & $8$ & $9$ & $9$ & $8$ \\
 & $9$ & $9$ & $9$ & $7$ & $8$ & $10$ & $10$ & $8$ & $10$ & $10$ \\ 
 & $9$ & $8$ & $8$ & $8$ & $7$ & $7$ & $7$ & $8$ & $7$ & $7$ \\ 
\hline 
\multirow{3}[2]{*}{$\Gamma(2)$} & $12$ & $13$ & $8$ & $\textbf{16}$ & $13$ & $13$ & $12$ & $12$ & $13$ & $12$ \\ 
 & $11$ & $12$ & $7$ & $\textbf{15}$ & $13$ & $14$ & $14$ & $11$ & $14$ & $14$ \\ 
 & $9$ & $10$ & $7$ & $\textbf{14}$ & $10$ & $10$ & $10$ & $10$ & $10$ & $10$ \\ 
\hline 
\multirow{3}[2]{*}{$\Gamma(3)$} & $15$ & $17$ & $11$ & $\textbf{20}$ & $19$ & $15$ & $12$ & $15$ & $15$ & $12$ \\ 
 & $15$ & $17$ & $10$ & $\textbf{20}$ & $18$ & $15$ & $15$ & $15$ & $15$ & $15$ \\ 
 & $11$ & $12$ & $7$ & $\textbf{17}$ & $12$ & $12$ & $11$ & $11$ & $11$ & $11$ \\
\hline 
\multirow{3}[2]{*}{$LN(0.5)$} & $40$ & $50$ & $33$ & $43$ & $\textbf{59}$ & $45$ & $29$ & $41$ & $45$ & $20$ \\ 
 & $39$ & $49$ & $29$ & $41$ & $\textbf{57}$ & $32$ & $22$ & $40$ & $32$ & $21$ \\ 
 & $23$ & $29$ & $16$ & $\textbf{36}$ & $30$ & $18$ & $17$ & $24$ & $18$ & $16$ \\  
\hline 
\multirow{3}[2]{*}{$LN(1)$} & $34$ & $\textbf{45}$ & $20$ & $39$ & $\textbf{45}$ & $27$ & $25$ & $34$ & $27$ & $24$ \\ 
 & $32$ & $\textbf{42}$ & $13$ & $23$ & $36$ & $30$ & $29$ & $28$ & $29$ & $27$ \\ 
 & $22$ & $28$ & $16$ & $\textbf{36}$ & $30$ & $18$ & $17$ & $23$ & $18$ & $16$ \\
\hline 
\multirow{3}[2]{*}{$\chi^2(8)$} & $17$ & $20$ & $14$ & $22$ & $\textbf{24}$ & $17$ & $13$ & $18$ & $18$ & $12$ \\ 
 & $17$ & $20$ & $13$ & $\textbf{23}$ & $\textbf{23}$ & $16$ & $14$ & $17$ & $16$ & $14$ \\ 
 & $12$ & $14$ & $8$ & $\textbf{20}$ & $13$ & $12$ & $12$ & $12$ & $12$ & $11$ \\
\hline 
\multirow{3}[2]{*}{$\chi^2(10)$} & $20$ & $23$ & $16$ & $24$ & $\textbf{27}$ & $22$ & $16$ & $19$ & $22$ & $13$ \\ 
 & $19$ & $23$ & $15$ & $25$ & $\textbf{26}$ & $19$ & $14$ & $19$ & $19$ & $14$ \\ 
 & $12$ & $14$ & $8$ & $\textbf{21}$ & $14$ & $13$ & $13$ & $13$ & $13$ & $12$ \\
\hline 
\multirow{3}[2]{*}{$\beta(1,1)$} & $55$ & $67$ & $65$ & $1$ & $\textbf{71}$ & $68$ & $57$ & $64$ & $69$ & $40$ \\ 
 & $54$ & $66$ & $64$ & $0$ & $\textbf{68}$ & $61$ & $46$ & $62$ & $62$ & $32$ \\ 
 & $10$ & $10$ & $\textbf{12}$ & $2$ & $5$ & $4$ & $4$ & $8$ & $4$ & $4$ \\
\hline 
\multirow{3}[2]{*}{$\beta(0.5,1)$} & $50$ & $59$ & $\textbf{60}$ & $0$ & $54$ & $28$ & $17$ & $55$ & $26$ & $16$ \\ 
 & $48$ & $54$ & $\textbf{59}$ & $0$ & $31$ & $3$ & $2$ & $48$ & $3$ & $2$ \\ 
 & $9$ & $\textbf{10}$ & $6$ & $5$ & $4$ & $5$ & $5$ & $8$ & $5$ & $5$ \\  
\hline 
\multirow{3}[2]{*}{$Lind(0.5)$} & $9$ & $10$ & $\textbf{12}$ & $7$ & $9$ & $9$ & $9$ & $8$ & $9$ & $9$ \\ 
 & $9$ & $9$ & $\textbf{11}$ & $5$ & $9$ & $8$ & $8$ & $8$ & $8$ & $8$ \\ 
& $9$ & $9$ & $\textbf{12}$ & $4$ & $7$ & $5$ & $5$ & $8$ & $5$ & $5$ \\ 
\hline 
\multirow{3}[2]{*}{$Lind(2)$} & $8$ & $8$ & $\textbf{9}$ & $6$ & $7$ & $7$ & $7$ & $8$ & $7$ & $7$ \\ 
 & $8$ & $8$ & $\textbf{9}$ & $4$ & $7$ & $7$ & $7$ & $7$ & $7$ & $7$ \\ 
 & $\textbf{8}$ & $\textbf{8}$ & $\textbf{8}$ & $7$ & $7$ & $7$ & $7$ & $\textbf{8}$ & $7$ & $7$ \\ 
 \hline \\[-1.8ex] 
\end{tabular} 
\end{table}

\begin{table}[!htbp] \centering 
  \caption{Estimated powers for $10\%$ censoring for a sample size of n=100 with three different censoring distributions.} 
  \label{table6} 
\begin{tabular}{@{\extracolsep{1pt}} ccccccccccc} 
\hline
$F$ & $KS$ & $CM$ & $LS$ & $KR$ & $S_{1}^{(1)}$ & $S_{5}^{(1)}$ & $S_{10}^{(1)}$ & $S_{1}^{(2)}$ & $S_{5}^{(2)}$ & $S_{10}^{(2)}$ \\
\hline 
\multirow{3}[2]{*}{$W(0.5)$} & $10$ & $10$ & $9$ & $7$ & $7$ & $7$ & $7$ & $8$ & $7$ & $7$ \\ 
 & $9$ & $9$ & $8$ & $3$ & $8$ & $9$ & $8$ & $8$ & $8$ & $8$ \\ 
 & $11$ & $12$ & $12$ & $7$ & $11$ & $11$ & $10$ & $11$ & $11$ & $10$ \\
\hline 
\multirow{3}[2]{*}{$W(1.5)$} & $10$ & $10$ & $10$ & $9$ & $9$ & $10$ & $9$ & $10$ & $10$ & $8$ \\ 
 & $10$ & $10$ & $10$ & $9$ & $9$ & $10$ & $9$ & $10$ & $10$ & $9$ \\ 
 & $10$ & $10$ & $9$ & $9$ & $8$ & $11$ & $10$ & $9$ & $11$ & $10$ \\ 
\hline 
\multirow{3}[2]{*}{$W(2)$} & $10$ & $10$ & $10$ & $9$ & $9$ & $9$ & $9$ & $10$ & $9$ & $8$ \\ 
 & $10$ & $10$ & $10$ & $10$ & $10$ & $10$ & $9$ & $10$ & $9$ & $8$ \\ 
 & $10$ & $10$ & $9$ & $9$ & $8$ & $11$ & $11$ & $9$ & $11$ & $11$ \\
\hline 
\multirow{3}[2]{*}{$\Gamma(2)$} & $17$ & $19$ & $14$ & $\textbf{32}$ & $22$ & $20$ & $16$ & $16$ & $20$ & $13$ \\ 
 & $16$ & $19$ & $14$ & $\textbf{30}$ & $21$ & $20$ & $16$ & $16$ & $20$ & $14$ \\ 
 & $11$ & $12$ & $9$ & $\textbf{19}$ & $11$ & $11$ & $11$ & $11$ & $12$ & $11$ \\  
\hline 
\multirow{3}[2]{*}{$\Gamma(3)$} & $25$ & $29$ & $22$ & $\textbf{45}$ & $36$ & $36$ & $31$ & $23$ & $36$ & $25$ \\ 
 & $24$ & $28$ & $21$ & $\textbf{44}$ & $35$ & $34$ & $28$ & $23$ & $34$ & $22$ \\ 
 & $14$ & $15$ & $11$ & $\textbf{27}$ & $16$ & $14$ & $12$ & $14$ & $14$ & $12$ \\ 
\hline 
\multirow{3}[2]{*}{$LN(0.5)$} & $74$ & $85$ & $71$ & $86$ & $92$ & $\textbf{94}$ & $91$ & $74$ & $\textbf{94}$ & $86$ \\ 
 & $73$ & $84$ & $70$ & $84$ & $92$ & $\textbf{93}$ & $91$ & $74$ & $\textbf{93}$ & $83$ \\ 
 & $38$ & $46$ & $33$ & $\textbf{63}$ & $56$ & $47$ & $33$ & $37$ & $47$ & $24$ \\ 
\hline 
\multirow{3}[2]{*}{$LN(1)$} & $71$ & $83$ & $61$ & $73$ & $\textbf{90}$ & $73$ & $37$ & $71$ & $71$ & $31$ \\ 
 & $70$ & $82$ & $53$ & $58$ & $\textbf{88}$ & $54$ & $39$ & $68$ & $53$ & $37$ \\ 
 & $38$ & $46$ & $33$ & $\textbf{63}$ & $57$ & $48$ & $33$ & $37$ & $47$ & $24$ \\
\hline 
\multirow{3}[2]{*}{$\chi^2(8)$} & $30$ & $36$ & $27$ & $\textbf{53}$ & $45$ & $47$ & $43$ & $28$ & $47$ & $36$ \\ 
 & $30$ & $36$ & $27$ & $\textbf{52}$ & $46$ & $46$ & $41$ & $29$ & $46$ & $34$ \\ 
 & $16$ & $18$ & $13$ & $\textbf{31}$ & $20$ & $16$ & $12$ & $16$ & $16$ & $12$ \\  
\hline 
\multirow{3}[2]{*}{$\chi^2(10)$}  & $34$ & $41$ & $32$ & $\textbf{57}$ & $51$ & $55$ & $52$ & $33$ & $55$ & $45$ \\ 
 & $34$ & $42$ & $32$ & $\textbf{59}$ & $52$ & $55$ & $51$ & $33$ & $55$ & $44$ \\ 
 & $18$ & $20$ & $14$ & $\textbf{34}$ & $22$ & $18$ & $14$ & $17$ & $18$ & $13$ \\  
\hline 
\multirow{3}[2]{*}{$\beta(1,1)$} & $91$ & $97$ & $94$ & $0$ & $98$ & $98$ & $98$ & $98$ & $\textbf{99}$ & $97$ \\ 
 & $91$ & $97$ & $94$ & $0$ & $98$ & $98$ & $97$ & $97$ & $\textbf{99}$ & $96$ \\ 
 & $\textbf{11}$ & $\textbf{11}$ & $\textbf{11}$ & $3$ & $9$ & $10$ & $9$ & $9$ & $9$ & $9$ \\ 
\hline 
\multirow{3}[2]{*}{$\beta(0.5,1)$} & $90$ & $96$ & $93$ & $0$ & $\textbf{98}$ & $97$ & $95$ & $97$ & $\textbf{98}$ & $91$ \\ 
 & $90$ & $96$ & $93$ & $0$ & $\textbf{98}$ & $97$ & $94$ & $97$ & $\textbf{98}$ & $88$ \\ 
 & $25$ & $\textbf{30}$ & $15$ & $24$ & $28$ & $13$ & $11$ & $21$ & $13$ & $11$ \\ 
\hline 
\multirow{3}[2]{*}{$Lind(0.5)$} & $11$ & $11$ & $\textbf{14}$ & $8$ & $11$ & $13$ & $11$ & $10$ & $12$ & $9$ \\ 
 & $11$ & $11$ & $\textbf{14}$ & $8$ & $11$ & $13$ & $11$ & $10$ & $12$ & $11$ \\ 
 & $12$ & $\textbf{13}$ & $\textbf{13}$ & $3$ & $11$ & $11$ & $9$ & $10$ & $10$ & $9$ \\ 
\hline 
\multirow{3}[2]{*}{$Lind(2)$} & $\textbf{10}$ & $\textbf{10}$ & $\textbf{10}$ & $6$ & $9$ & $\textbf{10}$ & $9$ & $9$ & $\textbf{10}$ & $9$ \\ 
 & $10$ & $10$ & $\textbf{11}$ & $6$ & $9$ & $\textbf{11}$ & $\textbf{11}$ & $9$ & $\textbf{11}$ & $\textbf{11}$ \\ 
 & $\textbf{10}$ & $\textbf{10}$ & $9$ & $8$ & $8$ & $\textbf{10}$ & $\textbf{10}$ & $9$ & $\textbf{10}$ & $\textbf{10}$ \\
\hline \\[-1.8ex] 
\end{tabular} 
\end{table}

\begin{table}[!htbp] \centering 
  \caption{Estimated powers for $20\%$ censoring for a sample size of n=100 with three different censoring distributions.} 
  \label{table7} 
\begin{tabular}{@{\extracolsep{1pt}} ccccccccccc} 
\hline
$F$ & $KS$ & $CM$ & $LS$ & $KR$ & $S_{1}^{(1)}$ & $S_{5}^{(1)}$ & $S_{10}^{(1)}$ & $S_{1}^{(2)}$ & $S_{5}^{(2)}$ & $S_{10}^{(2)}$ \\
\hline  
\multirow{3}[2]{*}{$W(0.5)$} & $8$ & $7$ & $7$ & $5$ & $7$ & $7$ & $7$ & $7$ & $7$ & $7$ \\ 
 & $9$ & $6$ & $8$ & $3$ & $7$ & $7$ & $7$ & $6$ & $7$ & $7$ \\ 
 & $9$ & $10$ & $9$ & $9$ & $8$ & $7$ & $7$ & $8$ & $7$ & $7$ \\  
\hline 
\multirow{3}[2]{*}{$W(1.5)$} & $9$ & $10$ & $10$ & $8$ & $8$ & $10$ & $10$ & $9$ & $10$ & $10$ \\ 
 & $10$ & $10$ & $9$ & $8$ & $8$ & $10$ & $10$ & $8$ & $10$ & $9$ \\ 
 & $9$ & $9$ & $9$ & $8$ & $7$ & $7$ & $7$ & $8$ & $7$ & $7$ \\ 
\hline 
\multirow{3}[2]{*}{$W(2)$} & $10$ & $10$ & $10$ & $8$ & $9$ & $9$ & $8$ & $9$ & $9$ & $8$ \\ 
 & $10$ & $10$ & $9$ & $8$ & $9$ & $10$ & $10$ & $9$ & $10$ & $10$ \\ 
 & $9$ & $10$ & $9$ & $8$ & $7$ & $7$ & $7$ & $8$ & $7$ & $7$ \\ 
\hline 
\multirow{3}[2]{*}{$\Gamma(2)$} & $15$ & $17$ & $12$ & $\textbf{25}$ & $18$ & $15$ & $13$ & $15$ & $15$ & $13$ \\ 
 & $15$ & $17$ & $10$ & $\textbf{21}$ & $16$ & $16$ & $15$ & $13$ & $15$ & $14$ \\ 
 & $12$ & $13$ & $9$ & $\textbf{19}$ & $11$ & $11$ & $11$ & $11$ & $11$ & $10$ \\  
\hline 
\multirow{3}[2]{*}{$\Gamma(3)$} & $23$ & $27$ & $19$ & $\textbf{33}$ & $32$ & $25$ & $17$ & $22$ & $25$ & $15$ \\ 
 & $21$ & $26$ & $15$ & $\textbf{31}$ & $29$ & $20$ & $17$ & $20$ & $20$ & $17$ \\ 
 & $15$ & $17$ & $11$ & $\textbf{26}$ & $16$ & $13$ & $13$ & $14$ & $13$ & $12$ \\
\hline 
\multirow{3}[2]{*}{$LN(0.5)$} & $67$ & $80$ & $61$ & $69$ & $\textbf{88}$ & $84$ & $71$ & $68$ & $84$ & $47$ \\ 
 & $66$ & $79$ & $55$ & $62$ & $\textbf{87}$ & $70$ & $39$ & $66$ & $69$ & $31$ \\ 
 & $41$ & $52$ & $31$ & $\textbf{57}$ & $\textbf{57}$ & $29$ & $24$ & $41$ & $29$ & $22$ \\
\hline 
\multirow{3}[2]{*}{$LN(1)$} & $61$ & $75$ & $43$ & $56$ & $\textbf{76}$ & $38$ & $32$ & $59$ & $37$ & $29$ \\ 
 & $\textbf{57}$ & $72$ & $31$ & $29$ & $\textbf{57}$ & $39$ & $36$ & $47$ & $37$ & $34$ \\ 
 & $40$ & $51$ & $31$ & $\textbf{58}$ & $56$ & $29$ & $23$ & $41$ & $28$ & $21$ \\ 
\hline 
\multirow{3}[2]{*}{$\chi^2(8)$} & $27$ & $32$ & $23$ & $39$ & $\textbf{40}$ & $35$ & $27$ & $26$ & $35$ & $19$ \\ 
 & $28$ & $33$ & $22$ & $38$ & $\textbf{39}$ & $28$ & $19$ & $26$ & $28$ & $17$ \\ 
 & $17$ & $20$ & $12$ & $\textbf{31}$ & $19$ & $14$ & $13$ & $17$ & $14$ & $13$ \\
\hline 
\multirow{3}[2]{*}{$\chi^2(10)$}  & $31$ & $38$ & $27$ & $42$ & $\textbf{47}$ & $44$ & $35$ & $30$ & $44$ & $25$ \\ 
 & $32$ & $38$ & $26$ & $42$ & $\textbf{46}$ & $38$ & $27$ & $30$ & $38$ & $20$ \\ 
 & $19$ & $22$ & $13$ & $\textbf{34}$ & $22$ & $15$ & $14$ & $18$ & $14$ & $13$ \\
\hline 
\multirow{3}[2]{*}{$\beta(1,1)$} & $87$ & $95$ & $91$ & $0$ & $\textbf{97}$ & $96$ & $94$ & $95$ & $\textbf{97}$ & $87$ \\ 
 & $87$ & $94$ & $90$ & $0$ & $\textbf{96}$ & $95$ & $92$ & $95$ & $\textbf{96}$ & $82$ \\ 
 & $14$ & $15$ & $\textbf{17}$ & $1$ & $6$ & $2$ & $2$ & $10$ & $2$ & $2$ \\
\hline 
\multirow{3}[2]{*}{$\beta(0.5,1)$} & $86$ & $91$ & $90$ & $0$ & $\textbf{93}$ & $87$ & $74$ & $92$ & $88$ & $67$ \\ 
 & $86$ & $89$ & $\textbf{91}$ & $0$ & $86$ & $55$ & $37$ & $87$ & $55$ & $35$ \\ 
 & $12$ & $\textbf{14}$ & $7$ & $3$ & $3$ & $4$ & $4$ & $9$ & $4$ & $4$ \\  
\hline 
\multirow{3}[2]{*}{$Lind(0.5)$} & $10$ & $11$ & $\textbf{13}$ & $7$ & $9$ & $9$ & $9$ & $8$ & $9$ & $9$ \\ 
 & $10$ & $11$ & $\textbf{12}$ & $5$ & $10$ & $8$ & $7$ & $8$ & $8$ & $7$ \\ 
 & $11$ & $12$ & $\textbf{14}$ & $3$ & $8$ & $5$ & $5$ & $8$ & $5$ & $5$ \\
\hline 
\multirow{3}[2]{*}{$Lind(2)$} & $9$ & $\textbf{10}$ & $\textbf{10}$ & $6$ & $7$ & $7$ & $7$ & $8$ & $7$ & $7$ \\ 
 & $9$ & $\textbf{10}$ & $9$ & $3$ & $7$ & $6$ & $6$ & $8$ & $6$ & $6$ \\ 
 & $\textbf{9}$ & $\textbf{9}$ & $\textbf{9}$ & $7$ & $7$ & $7$ & $7$ & $7$ & $7$ & $7$ \\
\hline \\[-1.8ex] 
\end{tabular} 
\end{table} 

\end{document}